\documentclass[11pt]{article}
\usepackage{authblk}
\usepackage{amssymb}
\usepackage{amsthm}
\usepackage{mathtools}
\usepackage{amsfonts}
\usepackage{amsmath}
\usepackage{graphicx} 
\usepackage{amsmath,amsthm,amssymb}
\usepackage[margin=1in]{geometry}
\usepackage[colorlinks=true]{hyperref}
\usepackage{thmtools, thm-restate}
\usepackage{algorithm}
\usepackage{algorithm}
\usepackage{algpseudocode}

\declaretheorem{theorem}
\usepackage{float}

\newtheorem{lemma}{Lemma}

\newtheorem{definition}{Definition}

\newtheorem*{definition*}{Definition}

\DeclareMathOperator{\Dim}{Dim}

\usepackage{tikz}
\usetikzlibrary{automata, positioning}
\usepackage{graphicx}
\title{A Markov-Chain Characterization of Finite-State Dimension and a Generalization of Agafonov’s Theorem}

\author{Laurent Bienvenu, Hugo Gimbert, Subin Pulari}

\date {\today}

\begin{document}

\maketitle

\begin{abstract}
	
Finite-state dimension quantifies the asymptotic rate of information in an infinite sequence as perceived by finite automata. For a fixed alphabet, the infinite sequences that have maximal finite-state dimension are exactly those that are Borel normal, i.e., in which all words of any given length appear with the same frequency. A theorem of Schnorr and Stimm (1972) shows that a real number is Borel normal if and only if, for every finite-state irreducible Markov chain with fair transitions, when the chain is simulated using the binary expansion of the given number, the empirical distribution of states converges to its stationary distribution. In this paper we extend this correspondence beyond normal numbers. We show that the finite-state dimension of a sequence can be characterized in terms of the conditional Kullback–Leibler divergence between the limiting distributions arising from the simulation of Markov chains using the given sequence and their stationary distributions. This provides a new information-theoretic characterization of finite-state dimension which generalizes the Schnorr–Stimm result.

As an application, we prove a generalization of Agafonov’s theorem for normal numbers. Agafonov’s theorem states that a sequence is normal if and only if every subsequence selected by a finite automaton is also normal. We extend this to arbitrary sequences by establishing a tight quantitative relationship between the finite-state dimension of a sequence and the finite-state dimensions of its automatic subsequences. 

\end{abstract}

\section{Introduction}

Finite-state dimension is a quantitative measure of the information content of an infinite sequence as observable by finite automata.  Introduced by Dai, Lathrop, Lutz, and Mayordomo~\cite{Dai2004}, it serves as a finite-state analogue of classical Hausdorff dimension. Originally defined in terms of winning rates of finite-state gamblers called $s$-gales, finite-state dimension can be characterized using a variety of computational and information-theoretic models, including finite-state predictors, compressors, and automatic Kolmogorov complexity~\cite{Dai2004,athreya2007effective,bourke2005entropy,Kozachinskiy2021}.  These equivalent formulations highlight the robustness of the notion and its close connections with information theory and automata-theory.

Given a finite alphabet~$\Sigma$, an infinite sequence $ X \in \Sigma^\infty $ is \emph{normal} if every block $ w \in \Sigma^k $ of length~$k$ occurs in $ X $ with limiting frequency $ |\Sigma|^{-k} $.  Normality captures randomness at the level of the occurrence frequencies of finite patterns within an infinite sequence.  A classical example is the Champernowne sequence $0123456789101112131415\dots$, which is normal over the decimal alphabet~$\{0,\ldots,9\}$.  It is conjectured that for any integer $b>2$, the base-$b$ expansions of natural constants such as~$\pi$,~$e$, and~$\sqrt{2}$ are normal in base~$b$, a longstanding open problem in mathematics.
Normal sequences are precisely those sequences having finite-state dimension equal to~1, that is, the maximal possible finite-state dimension~\cite{bourke2005entropy}.  In 1972, Schnorr and Stimm~\cite{Schnorr1972} established their landmark dichotomy theorem: for any sequence~$X$, if $X$ is non-normal, there exists a finite-state gambler that can win capital on~$X$ at an exponential rate, whereas if $X$ is normal, every finite-state gambler loses capital exponentially fast.  In the same work, they also provided a probabilistic perspective on normality, showing that a real number is normal if and only if, when its binary expansion drives any finite, irreducible Markov chain with fair transitions, the empirical distribution of visited states converges to the stationary distribution of that chain.

In this paper, for better readability, we will restrict our attention to binary infinite sequences, that is, when the alphabet $\Sigma$ is just $\{0,1\}$. However, all our results can be fairly easily generalized to arbitrary finite alphabets. 

Our first main result generalizes the Schnorr–Stimm characterization from normal sequences to sequences of arbitrary finite-state dimension. Given a sequence~$X$, we consider all finite-state \emph{fair} irreducible Markov chains~$M$ (those in which every state has two outgoing transitions labeled $0$ and $1$, each taken with probability $ \tfrac{1}{2} $), and we let~$M$ evolve according to the symbols of $ X $—that is, each symbol of $ X $ determines which labeled transition of $ M $ is followed at that step.  
For such an $ M $, let $ \mathcal{W}_M(X) $ denote the set of limiting empirical joint distributions of states and transitions observed during the run of $ M $ on $ X $.  
Our main theorem shows that the finite-state dimension of $ X $ can be expressed as  
\[
\dim_{\mathrm{FS}}(X)
= 1 - 
\sup_{M}
\sup_{\mu \in \mathcal{W}_M(X)}
D_{\mathrm{KL}}\!\big(\mu(E\mid Q) \,\|\, \pi_M(E \mid Q) \big),
\]
where the outer supremum is taken over all fair irreducible Markov chains~$M$, $\pi_M$ is the stationary distribution of~$M$, and $D_{\mathrm{KL}}$ denotes the Kullback--Leibler divergence.
We also provide an analogous characterization for the finite-state strong dimension $ \Dim_{\mathrm{FS}}(X) $.  
When $ \dim_{\mathrm{FS}}(X) = 1 $, the divergence term vanishes, recovering the Schnorr–Stimm lemma as a special case.  
Thus, our theorem extends the probabilistic characterization of normality due to Schnorr and Stimm to a full quantitative description of finite-state dimension for arbitrary sequences.

Our second main result concerns the behavior of finite-state dimension under finite-state selection.  
A \emph{finite-state selector} is a finite automaton equipped with a set of \emph{selecting states}: when run on an infinite sequence, it outputs the symbols read at those positions where the automaton enters a selecting state.  
Agafonov~\cite{Agafonov1968} proved that a sequence is normal if and only if every finite-state selector extracts from it a subsequence that is also normal. The relationship between subsequence selection and finite-state dimension is substantially more delicate. There exist simple examples, such as the \emph{diluted sequence} obtained by interleaving a normal sequence with zeros \cite{Dai2004}, whose finite-state dimension is only $\frac{1}{2}$, yet there are finite-state selectors that extract from it a completely non-random subsequence of zeros and others that extract a normal subsequence. 
We demonstrate the utility of our Markov-chain characterization to extend Agafonov’s classical normality result to all sequences of sufficiently high finite-state dimension, showing that the stability of randomness under finite-state selection persists beyond the case of normal numbers. For every irreducible finite-state selector~$\mathcal{S}$ with non-trivial set of selecting states (i.e., some but not all states are selecting states), we show that for every small error tolerance~$\varepsilon > 0$ there exists~$\delta > 0$ such that, whenever the finite-state dimension of~$X$ is at least~$1 - \delta$, the following inequality holds:
\[
\lambda_{\mathcal{S}}\dim_{\mathrm{FS}}(\mathcal{S}(X))
+(1-\lambda_{\mathcal{S}})\Dim_{\mathrm{FS}}(\mathcal{S}^{c}(X))
\ge \dim_{\mathrm{FS}}(X) - \varepsilon,
\]
where $\lambda_{\mathcal{S}} = \sum_{q \in S} \pi_{\mathcal{S}}(q)$
is the total stationary probability of the selecting states in the underlying Markov chain of~$\mathcal{S}$, $\mathcal{S}(X)$ is the subsequence of~$X$ selected by~$\mathcal{S}$ and $\mathcal{S}^{c}(X)$ is the subsequence \emph{not} selected by $\mathcal{S}$.

This result refines Agafonov’s selection theorem by showing that the total finite-state “information mass” of a sequence is approximately preserved under automatic selection.  
When $\dim_{\mathrm{FS}}(X)=1$, the inequality collapses to Agafonov’s theorem for normal numbers. We further demonstrate the utility of this generalization by deriving lower bounds on the finite-state dimensions of selected subsequences. Finally, we show that the bound is tight and analyze conditions under which the relationship fails to hold.

The connection between normality and finite-state dimension has been conjectured to extend far beyond these specific results.  
Lutz~\cite{LutzTalk2012} articulated this idea in what he called the \emph{Normality–Dimension Thesis}, which states that  
“every theorem about normality has a corresponding generalization in terms of finite-state dimension.”  
This thesis has been substantiated in several recent works showing that various theorems and statements concerning normality are special cases of more general results about finite-state dimension~\cite{Becher2025,fsdra,Gu2007,Nandakumar2025,Lutz2023,Kozachinskiy2021}.  
In the same spirit, our work provides further evidence for this thesis by demonstrating that both the Schnorr–Stimm lemma and Agafonov’s theorem emerge as special cases of broader quantitative principles concerning finite-state dimension.

\section{Preliminaries}
\label{sec:preliminaries}
Let $\{0,1\}^*$ denote the set of all finite binary strings, including the empty string~$\lambda$, and let $\{0,1\}^\infty$ denote the set of all infinite binary sequences.  
For $w \in \{0,1\}^*$, let $|w|$ denote its length.  
If $X = x_0 x_1 x_2 \ldots \in \{0,1\}^\infty$, then for each $n \in \mathbb{N}$, the prefix of $X$ of length~$n$ is written
$
X{\upharpoonright}n = x_0 x_1 \cdots x_{n-1}.
$. 
Given a finite set $\mathcal{X}$, we call $\mathsf{Prob}(\mathcal{X})$ the set of probability measures over~$\mathcal{X}$, which can be identified with the set of function $\pi: \mathcal{X} \rightarrow [0,1]$ such that $\sum_{x \in \mathcal{X}} \pi(x) =1$. We define $\|\pi\|_1$ and $\|\pi\|_\infty$ to be the quantities $\sum_{x \in \mathcal{X}} \pi(x)$ and $\max_{x \in \mathcal{X}} \pi(x)$ respectively. We also use $\|\cdot\|_{1}$ and $\|\cdot\|_{\infty}$ to denote the standard $\ell_{1}$ and $\ell_{\infty}$ norms on vectors and matrices. We consider the topology on the space $\mathsf{Prob}(\mathcal{X})$ induced by the distance $(\pi,\pi') \mapsto \| \pi - \pi' \|_1$ (which is the same topology as the one induced by the distance $(\pi,\pi') \mapsto \| \pi - \pi' \|_\infty$), which makes $\mathsf{Prob}(\mathcal{X})$ a compact metric space. Given a two distributions $\pi$ and $\pi'$ over $\mathcal{X}$, their Kullback-Leibler divergence (also known as relative entropy) is defined to be the quantity 
\[
D_{\mathrm{KL}} (\pi \| \pi') = \sum_{x \in \mathcal{X}} \pi(x) \log \frac{\pi(x)}{\pi'(x)}
\]
Which is a non-negative quantity and is equal to~$0$ if and only if $\pi=\pi'$. The Kullback-Liebler divergence relates to the $\|_1$-distance between distributions via the so-called Pinsker inequality:
\[
\| \pi-\pi'\|_1 \leq \sqrt{2 D_{\mathrm{KL}} (\pi \| \pi')}
\]

When $\pi$ and $\pi'$ are distributions over a product $\mathcal{X} \times \mathcal{Y}$ of finite sets, the conditional Kullback-Leibler divergence is the quantity
\begin{eqnarray*}
	D_{\mathrm{KL}} \big(\pi(\mathcal{Y}|\mathcal{X}) \| \pi'(\mathcal{Y}|\mathcal{X})) & = & \sum_{x \in \mathcal{X}} \pi(x)\cdot  D_{\mathrm{KL}} (\pi(.|x) \| \pi'(.|x))
	\\ 
	& = & \sum_{x \in \mathcal{X}} \pi(x) \sum_{y \in \mathcal{Y}} \pi(y|x) \log \frac{\pi(y|x)}{\pi'(y|x)}.
\end{eqnarray*}

In particular, we see from the first equality that if $x$ is a fixed element of~$\mathcal{X}$ with $\pi(x)>0$, one has by Pinsker's inequality
\[
\| \pi(.|x)-\pi'(.|x)\|_1 \leq \sqrt{2 D_{\mathrm{KL}} (\pi(.|x) \| \pi'(.|x))} \leq \sqrt{\frac{2 D_{\mathrm{KL}} (\pi(\mathcal{Y}|\mathcal{X}) \| \pi'(\mathcal{Y}|\mathcal{X}))}{\pi(x)}}.
\]

When $\pi$ is a distribution over a product $\mathcal{X} \times \mathcal{Y}$ of finite sets, 
the \emph{conditional entropy} of $\mathcal{X}$ given $\mathcal{Y}$ under $\pi$ is defined as
\[
H(\mathcal{X}_\pi\mid\mathcal{Y}_\pi)
= - \sum_{y \in \mathcal{Y}} \pi(y) \sum_{x \in \mathcal{X}} \pi(x \mid y) \log \pi(x \mid y).
\]
 An infinite sequence is normal if every finite block of symbols appears with the frequency expected under the uniform measure.

\begin{definition}[Normal sequence]
	A sequence $ X = x_0 x_1 x_2 \ldots \in \{0,1\}^\infty $ is \emph{normal} (in base~$2$) if, for every finite block $ w \in \{0,1\}^k $,
	\[
	\lim_{n \to \infty} \frac{\#\{\, 0 \le i < n : x_i x_{i+1} x_{i+2}..x_{i+k-1} = w \,\}}{n} = 2^{-k}.
	\]
	Equivalently, each $ k $-length block occurs in $ X $ with limiting frequency $ 2^{-k} $.
\end{definition}

A \emph{finite-state automaton} is a tuple $M = (Q, q_0, \delta)$, where $Q$ is a finite set of states, $q_0 \in Q$ is the initial state, and $\delta : Q \times \{0,1\} \to Q$ is a transition function.
A \emph{finite-state gambler} (or \emph{finite-state martingale}) is a finite automaton endowed with a betting function that specifies how the gambler bets on the next input symbol, given the finite prefix of the sequence observed so far.

\begin{definition}[Finite-state martingale \cite{Dai2004}]
	A \emph{finite-state martingale} (abbreviated \emph{FSM}) is a quadruple
	$D = (Q, q_0, \beta, \delta)$ where
	\begin{itemize}
		\item $Q$ is a finite set of states,
		\item $q_0 \in Q$ is the initial state,
		\item $\delta : Q \times \{0,1\} \to Q$ is the transition function, and
		\item $\beta : Q \to [0,1]$ assigns to each state $q$ a real number $\beta(q)$ representing the fraction of capital the gambler bets on the next symbol being $0$ (and $1-\beta(q)$ on the next symbol being $1$).
	\end{itemize}
	The capital of $D$ on a prefix $w = x_0 x_1 \cdots x_{n-1} \in \{0,1\}^*$ is defined recursively by
	\[
	D(\lambda) = 1, \qquad
	D(wx) =
	\begin{cases}
		D(w) \cdot 2\,\beta(q), & \text{if } x = 0,\\[3pt]
		D(w) \cdot 2\,(1-\beta(q)), & \text{if } x = 1,
	\end{cases}
	\]
	where $q$ is the state reached after reading $w$.
\end{definition}

\noindent
The factor~$2$ ensures that the expected capital remains constant under the uniform measure on~$\{0,1\}^\infty$.

\subsection{Finite-State Dimension}

Finite-state dimension quantifies the asymptotic rate of information in a sequence as perceived by finite automata~\cite{Dai2004}.  
It can be defined in terms of the exponential growth rate of finite-state martingales that succeed on~$X$.

\begin{definition}[Finite-state dimension and strong dimension~\cite{Dai2004,athreya2007effective}]
	For a sequence $ X \in \{0,1\}^\infty $, define
	\[
	\begin{aligned}
		\dim_{\mathrm{FS}}(X)
		&= \inf \bigl\{\, s \ge 0 : \exists \text{ FSM } D \text{ such that }
		\limsup_{n\to\infty} 2^{(s-1)n} D(X{\upharpoonright}n) = \infty \,\bigr\},\\[4pt]
		\Dim_{\mathrm{FS}}(X)
		&= \inf \bigl\{\, s \ge 0 : \exists \text{ FSM } D \text{ such that }
		\liminf_{n\to\infty} 2^{(s-1)n} D(X{\upharpoonright}n) = \infty \,\bigr\}.
	\end{aligned}
	\]
\end{definition}

Finite-state dimension was originally defined by Dai, Lathrop, Lutz, and Mayordomo~\cite{Dai2004} in terms of finite-state $s$-gales.  
In this paper, we adopt the equivalent martingale formulation, which is more convenient for our purposes.  
These notions satisfy $0 \le \dim_{\mathrm{FS}}(X) \le \Dim_{\mathrm{FS}}(X) \le 1$.  
An interesting special case arises when the finite-state dimension of a sequence equals one, which coincides exactly with the  normality~\cite{bourke2005entropy}.

\begin{theorem}[\cite{bourke2005entropy}]
	For every $ X \in \{0,1\}^\infty $,
	$
	\dim_{\mathrm{FS}}(X) = 1
	\text{ iff }
	X \text{ is normal.}
	$
\end{theorem}

\subsection{Finite-State Selectors}

A \emph{finite-state selector} is a deterministic finite automaton that, while reading an input sequence, outputs a subsequence consisting of the symbols read in designated ``selecting'' states.

\begin{definition}[Finite-state selector \cite{Agafonov1968}]
	A \emph{finite-state selector} is a tuple
	$\mathcal{S} = (Q, q_0, \delta, S)$ where
	$Q$ is a finite set of states,
	$q_0 \in Q$ is the start state,
	$\delta : Q \times \{0,1\} \to Q$ is the transition function,
	and $S \subseteq Q$ is the set of selecting states.
	Given $X = x_0 x_1 x_2 \cdots \in \{0,1\}^\infty$, the subsequence $\mathcal{S}(X)$ is obtained by outputting $x_i$ whenever the automaton is in a state $q_i \in S$ upon reading $x_i$.
\end{definition}

For a selector $\mathcal{S} = (Q, q_0, \delta, S)$, the \emph{complementary selector} $\mathcal{S}^c = (Q, q_0, \delta, Q \setminus S)$ selects exactly the symbols read in the nonselecting states.

\subsection{Markov Chains}
\label{sec:preliminariesmarkovchains}

A \emph{finite-state Markov chain} \cite{Chung1967} is a pair $(Q,P)$, where $Q$ is a finite set of states and $P = (p_{ij})_{i,j \in Q}$ is a stochastic transition matrix satisfying $p_{ij} \ge 0$ and $\sum_{j} p_{ij} = 1$ for all $i \in Q$.We call a finite-state Markov chain \emph{fair} if, from every state, there exist exactly two outgoing transitions with positive probability, each occurring with probability~$\tfrac{1}{2}$. For a probability vector $\pi$ on $Q$, $\pi$ is a \emph{stationary distribution} of $P$ if $\pi P = \pi$. For $ q,q' \in Q $, write $ q \to q' $ if there exists a finite path from $ q $ to $ q' $
with positive probability, and define $ q \leftrightarrow q' $ iff $ q \to q' $ and $ q' \to q $.
The equivalence classes of $ \leftrightarrow $ are the \emph{communication classes} of $ M $.

\begin{definition}[Ergodic set \cite{Chung1967}]
	Let $ M = (Q,P) $ be a finite-state Markov chain.  
	A subset $ E \subseteq Q $ is an \emph{ergodic set of $ M $} if  
	(i) for all $ q,q' \in E $, $ q \leftrightarrow q' $, and  
	(ii) there does not exist any $ q' \in Q \setminus E $ and $ q \in E $ such that $ q \to q' $.
\end{definition}

\begin{definition}[Irreducible Markov chain \cite{Chung1967}]
	A finite-state Markov chain $ M = (Q,P) $ is \emph{irreducible}
	if $M$ has a unique ergodic set $ E_M \subseteq Q $ and  
	for every $ q \in Q $, there exists some $ q' \in E_M $ with $ q \to q' $.\footnote{ In the standard terminology, a finite-state Markov chain is called irreducible if all its states belong to a single ergodic set. Here we adopt a slightly relaxed notion that allows a finite sequence of transient states before the chain enters its unique ergodic set.}
\end{definition}

Given a finite-state martingale or selector, we associate with it a Markov chain that captures only its state transitions under binary inputs.

\begin{definition}[Induced Markov chain]
	Let $ A = (Q, q_0, \delta) $ be a finite-state automaton (the case of interest being those underlying martingales or selectors).  
	The \emph{induced Markov chain} of $ A $ is $ M_A = (Q,q_0, E, P_A) $, where $ P_A $ is the stochastic matrix defined as follows:  
	for each state $ q \in Q $, there are two outgoing transitions labeled by $ 0 $ and $ 1 $, given by  
	$
	(q, 0, \delta(q,0)) \text{ and } (q, 1, \delta(q,1)),
	$
	each occurring with probability $ \tfrac{1}{2} $. The labels $ 0 $ and $ 1 $ correspond to the input symbols read by the automaton, and the transition $ (q,b,\delta(q,b)) $ represents moving from state $ q $ to $ \delta(q,b) $ upon reading symbol $ b $. $E$ denotes the set of all such transitions of $M_A$.
\end{definition}

\noindent
Let $\mathcal{A}$ be the set of all finite-state automata, and we define $\mathcal{M} = \{ M_A : A \in \mathcal{A} \text{ and } M_A \text{ is irreducible}.\}$ 
Throughout the paper, we work exclusively with the set $\mathcal{M}$ of irreducible Markov chains \emph{induced} by finite-state automata. For a state $q\in Q$, let $\mathsf{Out}(q)=\{(q,0,\delta(q,0)),(q,1,\delta(q,1))\}$ denote the set of outgoing transitions from~$q$. 
Given a sequence $X = x_0x_1x_2\cdots \in \{0,1\}^\infty$, the \emph{run of $M$ on $X$} is the sequence of states $q_0,q_1,q_2,\ldots$ defined by $q_{i+1}=\delta(q_i,x_i)$ for each $i\ge0$, starting from $q_0$. 
For each $ n \in \mathbb{N} $, let $ P_n(Q=q,E=e) $ denote the \emph{empirical joint probability} that, during the first $ n $ steps of the run of $ M $ on $ X $, the automaton is in state $ q $ and the next transition taken is $ e $:~\footnote{Note that if $ e \notin \mathsf{Out}(q) $, then $ P_n(Q=q,E=e)=0 $.}
\[
P_n(Q=q,E=e)
=\frac{1}{n}\bigl|\{\,0\le i<n : q_i=q\text{ and }(q_i,x_i,q_{i+1})=e\,\}\bigr|.
\]
The corresponding marginal distribution on states is 
$ Q_n(q) = \sum_{e \in \mathsf{Out}(q)} P_n(Q=q,E=e) $, 
and the conditional distribution of transitions given the current state is 
$ P_n(E=e \mid Q=q) = \tfrac{P_n(Q=q,E=e)}{Q_n(q)} $
whenever $ Q_n(q) > 0 $. We now define the set of limiting distributions associated with the run of $ M $ on $ X $.

\begin{definition}[Limiting distributions] The set of all limiting joint distributions over $ (Q,E) $ observed along the run of $ M $ on $ X $, denoted by $\mathcal{W}_M(X)$, is the closure of the sequence $(P_n)$ in the space of measures over $ Q\times E $, or, equivalently \[ \mathcal{W}_M(X) = \bigl\{\, \mu \in \mathsf{Prob}(Q\times E) : P_n \to \mu \text{ along some subsequence } n_k \bigr\}, \] where $ \mathsf{Prob}(Q\times E) $ denotes the space of probability measures on $ Q\times E $. \end{definition}
\noindent
Each such $ \mu $ satisfies $ \sum_{e \in \mathsf{Out}(q)} \mu(E=e \mid Q=q) = 1 $ for every $ q \in Q $ since they arise as limits of $P_n$.

\subsection{Irreducible Finite-State Martingales and Selectors}
\label{sec:irreducible_models}

We now consider \emph{irreducible} finite-state martingales and selectors, whose induced Markov chains are irreducible.  
These irreducible versions play a key technical role throughout the paper and we show that restricting attention to irreducible martingales or selectors does not alter the definitions or expressive power of the corresponding notions.

\begin{definition}[Irreducible finite-state martingale]
	A finite-state martingale $ D = (Q, q_0, \beta, \delta) $ is \emph{irreducible}
	if the Markov chain $ M_D = (Q, P_D) $ induced by its transition function $ \delta $
	is irreducible.  
\end{definition}
\begin{definition}[Irreducible finite-state dimension]
	The \emph{irreducible finite-state dimension} of a sequence $ X \in \{0,1\}^\infty $,
	denoted $ \dim_{\mathrm{FS}}^{\mathrm{irr}}(X) $,
	is the finite-state dimension obtained by restricting the definition
	of $ \dim_{\mathrm{FS}}(X) $ to irreducible finite-state martingales.
\end{definition}

\begin{lemma}
	\label{lem:irr_equiv}
	For every $ X \in \{0,1\}^\infty $,
	$
	\dim_{\mathrm{FS}}^{\mathrm{irr}}(X)
	= \dim_{\mathrm{FS}}(X).
	$
\end{lemma}

\begin{proof}
	It trivially follows that $\dim_{\mathrm{FS}}^{\mathrm{irr}}(X)
	\ge \dim_{\mathrm{FS}}(X)$. We show the converse. Let $ s > \dim_{\mathrm{FS}}(X) $.
	Then there exists a finite-state martingale $ D = (Q,q_0,\beta,\delta) $
	such that $ D(X{\upharpoonright}n) \ge 2^{(1-s)n} $ for infinitely many $ n $.
	During its run on $ X $, $ D $ eventually enters one of its ergodic sets,
	say $ E_D $, after reading some finite prefix $ w $, reaching a state $ q_w \in E_D $.
	
	We now construct an irreducible martingale $ D' = (Q',q_0',\beta',\delta') $
	that behaves exactly like $ D $ after processing the prefix $ w $.
	The new automaton $ D' $ begins in a fresh start state $ q_0' $
	and has a short initial path of transient states that simply reads the string $ w $.
	Along this prefix, $ D' $ follows the same sequence of transitions
	and places the same bets that $ D $ would make on $ w $.
	After the last symbol of $ w $ is read, $ D' $ enters the state $ q_w $
	and thereafter proceeds identically to $ D $, using the same betting and transition functions restricted to $ E_D $.
	
	The underlying Markov chain of $ D' $ is irreducible, since its only ergodic set is $ E_D $,
	and for all $ n \ge |w| $,
	$
	D'(X{\upharpoonright}n) = D(X{\upharpoonright}n).
	$
	Hence $ D' $ succeeds on $ X $ at the same rate as $ D $,
	and the two dimensions coincide.
\end{proof}

\noindent
Next, we define the corresponding notion of irreducibility for finite-state selectors, which will be crucial in our analysis of selection processes.

\begin{definition}[Irreducible selector]
	A finite-state selector $ \mathcal{S} = (Q, q_0, \delta, S) $ is \emph{irreducible}
	if the Markov chain $ M_{\mathcal{S}} = (Q, P_{\mathcal{S}}) $ induced by its transition function $ \delta $
	is irreducible.  
\end{definition}

\begin{lemma}
	\label{lem:irr_selector_equiv}
	For every sequence $ X \in \{0,1\}^\infty $ and every finite-state selector $ \mathcal{S} $
	that selects a subsequence $ \mathcal{S}(X) $ of $ X $,
	there exists an irreducible selector $ \mathcal{S}' $
	that selects exactly the same subsequence from $ X $.
\end{lemma}
\begin{proof}
	The forward implication is immediate:
	if an irreducible selector $ \mathcal{S}' $ selects a subsequence of $ X $,
	then so does $ \mathcal{S} = \mathcal{S}' $ itself.
	
	For the converse, let $ \mathcal{S} = (Q, q_0, \delta, S) $ be a (possibly reducible) selector
	that selects $ \mathcal{S}(X) $.
	During its run on $ X $, the underlying Markov chain $ M_{\mathcal{S}} = (Q, P_{\mathcal{S}}) $
	eventually enters one of its ergodic components, say $ E_{\mathcal{S}} $,
	after reading some finite prefix $ w \in \{0,1\}^* $.
	Let $ q_w $ denote the state reached after processing $ w $;
	from that point onward the run of $ \mathcal{S} $ remains entirely within $ E_{\mathcal{S}} $.
	
	We now construct an irreducible selector
	$ \mathcal{S}' = (Q', q_0', \delta', S') $
	that behaves identically to $ \mathcal{S} $ on $ X $ from that point onward.
	The new selector $ \mathcal{S}' $ begins in a fresh start state $ q_0' $
	and has a finite chain of transient states that simply reads the prefix $ w $
	and transitions to $ q_w $ in $ E_{\mathcal{S}} $.
	Along this prefix, $ \mathcal{S}' $ reproduces exactly the same sequence of selections
	that $ \mathcal{S} $ makes on $ w $.
	Once $ q_w $ is reached, the selector proceeds using the transition and selection structure
	of $ \mathcal{S} $ restricted to $ E_{\mathcal{S}} $.
	
	The underlying Markov chain of $ \mathcal{S}' $ is therefore irreducible,
	since its only ergodic component is $ E_{\mathcal{S}} $,
	and for all $ n \ge |w| $,
	$ \mathcal{S}'(X{\upharpoonright}n) = \mathcal{S}(X{\upharpoonright}n) $.
	Hence $ \mathcal{S}' $ selects exactly the same subsequence from $ X $,
	and the claim follows.
\end{proof}

\noindent
Thus, in both the gambling and selection settings, the restriction to irreducible finite-state automata entails no loss of generality, allowing us to work exclusively with irreducible models in subsequent sections.

\section{A Markov Chain Characterization of Finite-State Dimension}
\label{sec:markovchainfsd}

In this section, we establish a characterization of finite-state dimension in terms of Markov chains.  
Specifically, we express the finite-state dimension of a sequence~$X$ as a divergence-based quantity determined by the limiting empirical behavior of the irreducible Markov chains from the set~$\mathcal{M}$ when run on~$X$ (see Section~\ref{sec:preliminariesmarkovchains}).

\begin{theorem}
	\label{thm:markovchainfsd}
	For every $X \in \{0,1\}^\infty$,
	\[\dim_{\mathrm{FS}}(X)= 1 - 
	\sup_{M \in \mathcal{M}}\sup_{\mu \in \mathcal{W}_M(X)}
	D_{\mathrm{KL}}\!\big(\mu(E\mid Q) \,\|\, \pi_M(E \mid Q) \big).\]
\end{theorem}

\begin{proof}
	We first show that $\dim_{\mathrm{FS}}$ is less than or equal to the term on the right
	and then establish the converse inequality.

	\textit{(Forward direction.)}
	Let $ s $ be any real number strictly greater than 
	\[
	1 - 
	\sup_{M \in \mathcal{M}}\sup_{\mu \in \mathcal{W}_M(X)}
	D_{\mathrm{KL}}\!\big(\mu(E\mid Q) \,\|\, \pi_M(E \mid Q) \big)
	\]
	We will show that $ \dim_{\mathrm{FS}}(X) \le s $. Consider an irreducible Markov chain $ M $ and a limiting distribution 
	$ \mu \in \mathcal{W}_M(X) $ such that 
	$
	1 - D_{\mathrm{KL}}\!\big(\mu(E\mid Q) \,\|\, \pi_M(E \mid Q) \big) < s.
	$
	Define a martingale $ D $ that, when the automaton is in state $ q $, bets according to the conditional distribution $ \mu(E=e\mid Q=q) $. For any distribution $P$ on $E \times Q$ we use $P(e \mid q)$ to denote $P(E=e \mid Q=q)$. For the run of $M$ on $X$, recall the definition of the empirical joint distribution $P_n$ and its corresponding marginal distribution $Q_n$ from section \ref{sec:preliminaries}.
	
	Now choose a subsequence $ \langle n_m\rangle $ along which $ P_{n_m} \to \mu $.
	The capital of $ D $ after reading $ X{\upharpoonright} n_m $ is
	\[
	D(X{\upharpoonright} n_m)
	= 2^{\,n_m}
	\prod_{q \in Q}\prod_{e \in E}
	\bigl(\mu(E=e\mid Q=q)\bigr)^{P_{n_m}(E=e\mid Q=q)\,Q_{n_m}(q)}.
	\]
	
	Taking logarithms gives
	\[
	\log D(X{\upharpoonright} n_m)
	= n_m - n_m \sum_{q \in Q} Q_{n_m}(q)
	\sum_{e \in E} P_{n_m}(e\mid q)
	\log \frac{1}{\mu(E=e\mid Q=q)}.
	\]
	For sufficiently large $ n_m $, the empirical distributions are close to their limiting values, so using the continuity of conditional entropy \cite{Cover1991},
	\begin{align*}
		&\sum_{q \in Q} Q_{n_m}(q)
		\sum_{e \in E} P_{n_m}(e\mid q)
		\log \frac{1}{\mu(e\mid q)} 
		= \sum_{q \in Q} Q_{\mu}(q)
		\sum_{e \in E} \mu(e\mid q)
		\log \frac{1}{\mu(e\mid q)} + \varepsilon_m,
	\end{align*}
	where $ \varepsilon_m\to 0 $ as $ m\to\infty $.
	
	Because each transition of $ M $ occurs with probability $ 1/2 $, for every state $ q\in Q $ we have
	\[
	\pi_M(E=e\mid Q=q)=\tfrac{1}{2}\quad\text{for each }e\in E,
	\]
	and therefore
	\begin{align*}
		D_{\mathrm{KL}}\!\left(E_\mu\mid Q_\mu\,\middle\|\,E_{\pi_M}\mid Q_{\pi_M}\right)
		&= \sum_{q\in Q}Q_{\mu}(q)\sum_{e\in E}\mu(e\mid q)\log\frac{\mu(e\mid q)}{1/2}\\
		&= 1-H(E_\mu\mid Q_\mu),
	\end{align*}
	where $ H(E_\mu\mid Q_\mu) $ is the conditional entropy. Since
	$
	D_{\mathrm{KL}}\!\big(\mu(E\mid Q) \,\|\, \pi_M(E \mid Q) \big) > 1-s
	$
	by assumption, for all sufficiently large $ m $,
	\[
	\frac{1}{n_m}\log D(X{\upharpoonright} n_m) \ge 1 - s - \varepsilon_m.
	\]
	Hence the martingale $ D $ wins at rate at least $ 2^{(1-s-\varepsilon_m)n_m} $,
	which shows that $ \dim_{\mathrm{FS}}(X)\le s $.

	\textit{(Converse direction.)}
	Suppose now that $ s > \dim_{\mathrm{FS}}(X) $.
	Then there exists an irreducible finite-state martingale $ D $ that succeeds on $ X $ at rate $ 2^{(1-s)n} $; that is, for infinitely many $ n\in\mathbb{N} $,
	$
	D(X{\upharpoonright} n) \ge 2^{(1-s)n}.
	$
	From this martingale $ D $, construct an irreducible labeled Markov chain $ M_D $ having the same set of states as $ D $.
	Each state $ q\in Q $ has two outgoing transitions labeled $ 0 $ and $ 1 $, each occurring with probability $ 1/2 $.
	This defines a fair irreducible Markov chain $ M_D $ that mirrors the structure of the martingale.
	
	Let $P_n$ denote the empirical joint probability distribution for the run of $M_D$ on $X$ and let $ Q_n$ denote the corresponding marginal distribution over the states (see section \ref{sec:preliminaries}).  
	Then
	\[
	D(X{\upharpoonright} n)
	= 2^{\,n}\prod_{q\in Q}\prod_{e\in E}\bigl(P_n(E=e\mid Q=q)\bigr)^{P_n(E=e\mid Q=q)\,Q_n(q)}.
	\]
	Taking logarithms and dividing by $ n $ gives
	\[
	\frac{1}{n}\log D(X{\upharpoonright} n)
	= 1 - \sum_{q\in Q}Q_n(q)\sum_{e\in E}P_n(E=e\mid Q=q)\log\frac{1}{P_n(E=e\mid Q=q)}.
	\]
	For infinitely many $ n $, we have $ D(X{\upharpoonright} n)\ge 2^{(1-s)n} $, implying
	\[
	1 - \sum_{q\in Q}Q_n(q)\sum_{e\in E}P_n(E=e\mid Q=q)\log\frac{1}{P_n(E=e\mid Q=q)} \ge 1 - s.
	\]
	Since the space of distributions is compact, we may choose a cluster point $ \mu $ of a subsequence of  $ \{P_n\}_{n \geq 1} $ satisfying the above inequality.
	By continuity of conditional entropy,
	\[
	1 - \sum_{q\in Q}Q_{\mu}(q)\sum_{e\in E}\mu(E=e\mid Q=q)\log\frac{1}{\mu(E=e\mid Q=q)} \ge 1 - s.
	\]
	For the fair irreducible Markov chain $ M_D $, we again have $ \pi_{M_D}(E=e\mid Q=q)=1/2 $, so
	\begin{align*}
		D_{\mathrm{KL}}\!\big(\mu(E\mid Q) \,\|\, \pi_{M_D}(E \mid Q) \big)
		&= \sum_{q\in Q}Q_{\mu}(q)\sum_{e\in E}\mu(e\mid q)\log\frac{\mu(e\mid q)}{1/2}\\
		&= 1-H(E_\mu\mid Q_\mu).
	\end{align*}
	Hence,
	$
	1 - D_{\mathrm{KL}}\!\big(\mu(E\mid Q) \,\|\, \pi_{M_D}(E \mid Q) \big) \le s.
	$
	It follows that
	\[
	1 - 
	\sup_{M \in \mathcal{M}}\sup_{\mu \in \mathcal{W}_M(X)}
	D_{\mathrm{KL}}\!\big(\mu(E\mid Q) \,\|\, \pi_M(E \mid Q) \big)
	\le s.
	\]
	Since $ s > \dim_{\mathrm{FS}}(X) $ was arbitrary, this establishes the converse inequality and completes the proof.
\end{proof}

We next extend this characterization to  finite-state strong dimension.

\begin{theorem}
	\label{thm:markovchainstrongfsd}
	For every sequence $X \in \{0,1\}^\infty$,
	\[
	\Dim_{\mathrm{FS}}(X)= 
	1 - 
	\sup_{M \in \mathcal{M}}\inf_{\mu \in \mathcal{W}_M(X)}
	D_{\mathrm{KL}}\!\big(\mu(E\mid Q) \,\|\, \pi_M(E \mid Q) \big).
	\]
\end{theorem}

\begin{proof}
	The structure of the proof parallels that of the previous theorem, and we only indicate the essential differences.

	\textit{(Forward direction.)}
	Let $ s $ be any real number strictly greater than the term on the right. We show that $ \Dim_{\mathrm{FS}}(X) \le s $. Choose an irreducible Markov chain $M$ such that
	$
	1 - 
	\inf_{\mu\in \mathcal{W}_M(X)}
	D_{\mathrm{KL}}\!\big(\mu(E\mid Q) \,\|\, \pi_M(E \mid Q) \big) < s.
	$
	Hence for every limiting distribution $\mu\in\mathcal{W}_M(X)$,
	$
	D_{\mathrm{KL}}\!\big(\mu(E\mid Q) \,\|\, \pi_M(E \mid Q) \big)\ge 1-s.
	$
	Let
	\[
	\mathcal{C}=\Bigl\{\mu\in\mathcal{W}_M(X):
	D_{\mathrm{KL}}\!\big(\mu(E\mid Q) \,\|\, \pi_M(E \mid Q) \big)\ge 1-s\Bigr\}.
	\]
	By continuity of $D_{\mathrm{KL}}$, $\mathcal{C}$ is a closed subset of the compact space of probability measures on $Q\times E$; hence $\mathcal{C}$ is compact.  
	Cover $\mathcal{C}$ by finitely many balls $B_1,\dots,B_N$ of small radius $\varepsilon'>0$ centered at measures
	$\mu_1,\dots,\mu_N\in\mathcal{C}$.
	
	For each center $\mu_i$, define a finite-state martingale $D_i$ that bets according to $\mu_i(E=e\mid Q=q)$, exactly as in the first part of the proof of the previous theorem.  
	Each $D_i$ wins on any sequence whose empirical distribution is sufficiently close to $\mu_i$ at rate at least $2^{(1-s)n}$ up to an error that vanishes with $\varepsilon'$.
	
	Now, since every limiting distribution of the run of $M$ on $X$ lies in $\mathcal{C}$,
	for all sufficiently large $n$ the empirical measure of transitions lies within one of the balls $B_i$.
	When this happens, the empirical distribution of transitions differs from the center $\mu_i$ by at most $\varepsilon'$ in total variation.
	By continuity of $D_{\mathrm{KL}}$, the value
	$
	D_{\mathrm{KL}}\!\big(P_n(E\mid Q) \,\|\, \pi_M(E \mid Q) \big)
	$
	differs from 
	$
	D_{\mathrm{KL}}\!\big(\mu_i(E\mid Q) \,\|\, \pi_M(E \mid Q) \big)
	$
	by at most an arbitrarily small amount.
	Hence, for all sufficiently large $n$, one of the martingales $D_i$ is winning at rate at least $2^{(1-s-\varepsilon)n}$ for arbitrarily $\varepsilon>0$.
	
	The key to finish the proof is to consider the collection $\{D_1,\dots,D_N\}$ as an $N$-account martingale.  
	By the near-simulation of multi-account by single-account finite-state gamblers (see \cite[Theorem 4.5]{Dai2004}), this implies the existence of a \emph{single} finite-state martingale $D$ that succeeds on $X$ at rate at least $2^{(1-s-2\varepsilon)n}$.  
	Since $\varepsilon$ is arbitrarily small, we have $\Dim_{\mathrm{FS}}(X)\le s.$

	\textit{(Converse direction.)}
	Assume $s>\Dim_{\mathrm{FS}}(X)$.  
	Then there exists an irreducible FSM $D$ that strongly succeeds on $X$ at rate $2^{(1-s)n}$ eventually.  
	Exactly as in the converse direction of the previous theorem, define the corresponding irreducible Markov chain $M_D$ and the empirical measures $P_n$.  
	Since the success is strong (i.e., holds for all sufficiently large $n$), every cluster point $\mu$ of the sequence $\{P_n\}$ must satisfy
	$
	D_{\mathrm{KL}}\!\big(\mu(E\mid Q) \,\|\, \pi_{M_D}(E \mid Q) \big)\ge 1-s.
	$
	Because this holds for every limiting measure of $M_D$ on $X$, this establishes the converse.
\end{proof}

We now show that Lemma~4.5 from Schnorr and Stimm~\cite{Schnorr1972},  
which characterizes normality in terms of empirical state convergence in fair Markov chains,  
is a special case of our Markov-chain characterization of finite-state dimension.

\begin{lemma}[Schnorr and Stimm \cite{Schnorr1972}]
	A binary sequence $X \in \{0,1\}^\omega$ is \emph{normal} if and only if, for every irreducible finite-state Markov chain $M$ in which each state has two outgoing transitions labeled $0$ and $1$, each taken with probability $\tfrac{1}{2}$, the empirical distribution of the states visited by $M$ during its run on $X$ converges to the stationary distribution $\pi_M$ of $M$.
\end{lemma}

\begin{proof}
	(\emph{Forward direction.})
	If $ X $ is normal, then $ \dim_{\mathrm{FS}}(X)=1 $.  
	From the characterization theorem above, we immediately obtain that for every fair irreducible Markov chain $ M $ and every $ \mu \in \mathcal{W}_M(X) $,
	$
	D_{\mathrm{KL}}\!\big(\mu(E\mid Q) \,\|\, \pi_{M}(E \mid Q) \big)=0.
	$
	Hence $ \mu(E=e\mid Q=q)=\pi_M(E=e\mid Q=q) $ for all $ e,q $, and every state $ q $ has $\mu(Q=q)>0$, since in an irreducible fair Markov chain with a single ergodic component the probability mass cannot remain confined to a proper subset of the ergodic component. Since $ M $ is fair, this means
	$
	\mu(E=e\mid Q=q)=\tfrac{1}{2}\quad\text{for all }e,q.
	$
	Let $ P $ denote the transition matrix of $ M $.  
	Then the marginal distribution $ Q_\mu $ on states satisfies $ Q_\mu P = Q_\mu $, which is exactly the condition defining the stationary distribution $ \pi_M $, i.e., $ \pi_M P = \pi_M $ and $ \sum_q \pi_M(q)=1 $. For an irreducible finite-state Markov chain, the stationary distribution $ \pi_M $ is unique, so $ Q_\mu=\pi_M $.  
	Since $ \mu(E=e\mid Q=q)=\pi_M(E=e\mid Q=q) $ and $ Q_\mu=\pi_M $, we have $ \mu=\pi_M $.  
	Therefore, along the run of $ M $ on $ X $, the empirical distribution of the visited states converges to $ \pi_M $.

	(\emph{Converse direction.})
	The converse is straightforward.  
	If $ X $ is not normal, then for some block length $ L $ the frequencies of $ L $-bit blocks in $ X $ do not converge to the uniform distribution.  
	Consider the finite automaton (equivalently, a Markov chain) that remembers the last $ L $ bits of the sequence.  
	This chain is fair and has the uniform stationary distribution on $ \{0,1\}^L $.  
	When run on $ X $, the empirical distribution of its states equals the empirical distribution of $ L $-bit blocks, which by assumption fails to converge to uniform.  
	Hence, for this chain, the empirical distribution of states does not converge to $ \pi_M $.
\end{proof}

\section{Agafonov's Theorem and Finite-State Dimension}
\label{sec:agafonov}

The relationship between normality and finite-state selection has its origins in a classical result of Agafonov~\cite{Agafonov1968}, who showed that normality is preserved under all finite-state selection procedures. In other words, a sequence remains normal even when a finite automaton selects a subsequence from it. Formally:

\begin{theorem}[Agafonov \cite{Agafonov1968}]
	\label{thm:agafonovfornormal}
	A real number $ x $ is normal if and only if the subsequence of $ x $
	selected by any finite-state selector is normal.
\end{theorem}

When we turn from normality to finite-state dimension, the connection between subsequence selection and randomness becomes considerably more subtle.  Simple examples already show that the behavior can vary dramatically: for instance, in the \emph{diluted sequence} $0r_10r_20r_30r_4\dots$ obtained by interleaving a normal sequence with zeros \cite{Dai2004}, the overall dimension is $\tfrac{1}{2}$, yet there exist simple automatic selectors that extract from it one subsequence consisting entirely of zeros (of dimension~$0$) and another that reproduces the original normal sequence (of dimension~$1$).  Thus, in general, finite-state selection need not preserve dimension in any obvious sense.  Nevertheless, by employing the Markov chain characterization developed in Section~\ref{sec:markovchainfsd}, we show that an analogue of Agafonov’s stability phenomenon continues to hold for all sufficiently random sequences: for every irreducible finite-state selector, the finite-state dimensions of the selected and unselected subsequences remain quantitatively balanced whenever the original sequence has high dimension. Later in this section, we show that Agafonov’s theorem for normal numbers is a special case of this more general result. The following technical lemmas will be used in the proof of the main theorem.

\begin{lemma}
	\label{lem:convergenceofstationarydistributions}
	Let $ M $ and $ M' $ be two finite-state irreducible Markov chains with transition matrices 
	$ P $ and $ P' $, and stationary distributions $ \pi $ and $ \pi' $, respectively.  
	For every $ \varepsilon > 0 $, there exists $ \delta > 0 $ such that if
	$
	\| P - P' \|_{\infty} < \delta,
	$
	then
	$
	\| \pi - \pi' \|_{1} < \varepsilon.
	$
\end{lemma}

\begin{proof}
	The stationary distribution $ \pi $ of a finite irreducible Markov chain with transition matrix $ P $ is the unique solution of the linear system consisting of $\pi (I - P) = 0$ and $\pi \mathbf{1} = 1$,
	where $ \mathbf{1} $ denotes the all-ones column vector. These two equations can be combined into a single linear system
	$
	\pi M = b,
	$
	where $ M $ is obtained from $ I - P $ by replacing one of its rows with $ \mathbf{1}^\top $, and $ b $ is the vector having a $ 1 $ in the corresponding position and $ 0 $ elsewhere.  
	Since $ P $ is irreducible, $ M $ is invertible, and thus
	$
	\pi = b M^{-1}.
	$
	
	Now let $ P' $ be another transition matrix, and define $ M' $ and $ \pi' $ analogously.  
	Both $ M $ and $ M' $ belong to the open set
	$
	\mathrm{GL}(n, \mathbb{R}) = \{ M \in \mathbb{R}^{n \times n} : \det(M) \neq 0 \}.
	$
	On this set, the mapping $ M \mapsto M^{-1} $ is continuous, since the entries of $ M^{-1} $ are rational functions of the entries of $ M $ and $ \det(M) $ does not vanish.  
	It follows that the stationary distribution $ \pi = b M^{-1} $ depends continuously on $ P $, and the theorem follows immediately from this observation.
\end{proof}

\begin{lemma} \label{lem:min-frequency}
	Let $M$ be an irreducible fair Markov chain. For any $\eta >0$, there exists a $\delta >0$ such that if $\dim_{\mathrm{FS}}(X) > 1-\delta$, then for any $\mu \in \mathcal{W}_M(X)$, and any $q$ in the ergodic component of $M$, $\mu(q)> (1-\eta)\pi_M(q)$. 
\end{lemma}

\begin{proof}
	Given $\eta$, let us fix a small $\varepsilon$, to be specified later. By ergodicity of the Markov chain and the definition of the invariant measure, there exists an $N$ such that, if we run the Markov chain during $N$ steps on a random $u \in \{0,1\}^N$, with probability $>1-\varepsilon$, each state~$q$ of the ergodic component of $M$ gets visited at least $(1-\varepsilon)\pi_M(q) N$ times, regardless of the starting state. Let $U$ be the set of words~$u$ which make this happen (hence $|U| \geq (1-\varepsilon) 2^N$). 
	
	As proven in~\cite[Theorem 6.1]{bourke2005entropy}, for $\delta$ small enough, if $\dim_{\mathrm{FS}}(X) > 1-\delta$, each word of length~$N$ must appear in~$X$ is min-frequency at least $(1-\varepsilon) 2^N$. Thus, for such a $\delta$, words of $U$ appear in~$X$ with min-frequency at least $(1-\varepsilon)^2$.  By definition of $U$, this means that in the run of $M$ on any $X$ with $\dim_{\mathrm{FS}}(X) > 1-\delta$, each state~$q$ of the ergodic component of $M$ gets visited with min-frequency at least $(1-\varepsilon)^3 \pi_M(q)$. Thus, taking $\varepsilon$ so that $(1-\varepsilon)^3 = (1-\eta)$ and the corresponding~$\delta$, we get the desired result. 
\end{proof}

We now present the main result of this section, which generalizes Agafonov’s theorem from normality to finite-state dimension.

\begin{theorem}
	\label{thm:agafonovforfsd}
	For every irreducible selector $ \mathcal{S} =(Q, q_0, \delta, S)$, there exists a constant 
	$ \lambda_{\mathcal{S}} \in (0,1) $ such that for every $ \varepsilon > 0 $  
	there exists $ \delta > 0 $ satisfying the following:  
	for every sequence $ X $ with $ \dim_{\mathrm{FS}}(X) \ge 1 - \delta $,
	\[
	\lambda_{\mathcal{S}} \cdot \dim_{\mathrm{FS}}(\mathcal{S}(X))
	+ (1 - \lambda_{\mathcal{S}}) \cdot \Dim_{\mathrm{FS}}(\mathcal{S}^{c}(X))
	\ge \dim_{\mathrm{FS}}(X) - \varepsilon,
	\]
	where
	$
	\lambda_{\mathcal{S}} = \sum_{q \in S} \pi_{\mathcal{S}}(q)
	$
	is the total stationary probability of the selecting states under the stationary distribution  
	$ \pi_{\mathcal{S}} $ of $ M_{\mathcal{S}} $
	(and thus $ 0 < \lambda_{\mathcal{S}} < 1 $ whenever the ergodic component of $ M_{\mathcal{S}} $ 
	contains both selecting and non-selecting states).
\end{theorem}

\begin{proof}
	Let $ \mathcal{S} $ be a fixed selector and let $ M_{\mathcal{S}} $ denote its corresponding Markov chain. Let $ \lambda_{\mathcal{S}} = \sum_{q \in S} \pi_{\mathcal{S}}(q) $ denote the total stationary probability of the selecting states under the stationary distribution $ \pi_{\mathcal{S}} $ of $ M_{\mathcal{S}} $. Assume that the sequence $ X $ satisfies $ \dim_{\mathrm{FS}}(X) \ge 1 - \delta $ for some sufficiently small $ \delta > 0 $. From the Markov-chain characterization of finite-state dimension established earlier, for every limiting distribution $ \mu \in \mathcal{W}_{M_{\mathcal{S}}}(X) $ we have $ D_{\mathrm{KL}}\!\left(E_\mu\mid Q_\mu\,\middle\|\,E_{\pi_{\mathcal{S}}}\mid Q_{\pi_{\mathcal{S}}}\right) < \delta$. Using the Pinsker’s inequality, for every state $ q $, $ \| \mu(\cdot \mid q) - \pi_{\mathcal{S}}(\cdot \mid q) \|_{1} \le \sqrt{2\delta/\mu(q) }$. 
	
	Using Lemma~\ref{lem:min-frequency}, as long as $\delta$ is small enough, we have that $\mu(q)$ is at least $c/2$ where $c>0$ is the minimum value of $\pi_{\mathrm{S}}$ over all states. Thus, the above inequality gives 
	\[
	\| \mu(\cdot \mid q) - \pi_{\mathcal{S}}(\cdot \mid q) \|_{1} \le 2 \sqrt{\delta/c }
	\]
	where now the upper bound does not depend on~$q$. 
	
	Since each of these conditional distributions has support on two symbols (the outgoing edges labeled $0$ and $1$), the $ L_{\infty} $ deviation between them satisfies \[ \| \mu(\cdot \mid q) - \pi_{\mathcal{S}}(\cdot \mid q) \|_{\infty} = \tfrac{1}{2} \| \mu(\cdot \mid q) - \pi_{\mathcal{S}}(\cdot \mid q) \|_{1} \le \sqrt{\delta/c} . \] Hence, the transition matrices $ P_{\mu} $ and $ P_{\mathcal{S}} $ corresponding respectively to these conditional probabilities differ in each entry by at most $ \Delta = \sqrt{\delta/c}. $ 
	Fix such a limiting distribution $ \mu \in \mathcal{W}_{M_{\mathcal{S}}}(X) $ and consider the Markov chain $ M_{\mu} $ whose transition matrix is $ P_{\mu} $. Since $ D_{\mathrm{KL}}\!\left(E_\mu\mid Q_\mu\,\middle\|\,E_{\pi_{\mathcal{S}}}\mid Q_{\pi_{\mathcal{S}}}\right) < \delta $, the matrices $ P_{\mu} $ and $ P_{\mathcal{S}} $ are $ \Delta $-close in every entry. By Lemma \ref{lem:convergenceofstationarydistributions}, for every $ \varepsilon > 0 $ there exists a $ \Delta_{0} > 0 $ depending only on $ M_{\mathcal{S}} $ such that if $\|P_{\mu} - P_{\mathcal{S}}\|_{\infty} < \Delta_{0}$, then the stationary distributions $ \pi_{\mu} $ and $ \pi_{\mathcal{S}} $ satisfy $ \| \pi_{\mu} - \pi_{\mathcal{S}} \|_{1} < \varepsilon. $ 
	Choosing $ \delta $ sufficiently small so that $ \Delta = \sqrt{\delta/c}  < \Delta_{0} $, we conclude that the stationary distribution of $ M_{\mu} $ is within $ \varepsilon $ of $ \pi_{\mathcal{S}} $ in the $ L_{1} $-norm. By construction of $ M_{\mu} $, the stationary distribution of $ M_{\mu} $ is precisely the marginal $ Q_{\mu} $ on the states induced by $ \mu $. Hence, $ \| Q_{\mu} - \pi_{\mathcal{S}} \|_{1} < \varepsilon$. It follows that for all sufficiently large $ n $, the empirical frequency $ \frac{1}{n} N_{\mathrm{sel}}(n) $ of visits to selecting states during the run of $ M_{\mathcal{S}} $ on $ X $ satisfies \[ \Bigl|\tfrac{1}{n} N_{\mathrm{sel}}(n) - \lambda_{\mathcal{S}}\Bigr| < \varepsilon, \] where $ N_{\mathrm{sel}}(n) $ denotes the number of positions among the first $ n $ steps in which the selector is in a selecting state. We now construct the martingale that witnesses the required inequality. 
	Let $ s > \dim_{\mathrm{FS}}(\mathcal{S}(X)) $ and $ s' > \Dim_{\mathrm{FS}}(\mathcal{S}^{c}(X)) $. By the definitions of finite-state and strong finite-state dimension, there exist martingales $ G_{1} $ and $ G_{2} $ such that $ G_{1}(\mathcal{S}(X){\upharpoonright}n) \ge 2^{(1-s)n} $ for infinitely many $ n $, and $ G_{2}(\mathcal{S}^{c}(X){\upharpoonright}n) \ge 2^{(1-s')n} $ for all but finitely many $ n $. We combine $ G_{1} $ and $ G_{2} $ with the selector automaton $ \mathcal{S} $ to define a new martingale $ G $. The state of $ G $ is the triple consisting of the current states of $ G_{1} $, $ G_{2} $, and $ \mathcal{S} $. At each step, suppose the selector automaton $ \mathcal{S} $ is in state $ q $. If $ q $ is a selecting state, let $ p_{1,q} $ denote the betting distribution prescribed by $ G_{1} $ in its current state; if $ q $ is a non-selecting state, let $ p_{2,q} $ denote the betting distribution prescribed by $ G_{2} $ in its current state. The combined martingale $ G $ places its bet according to \[ p_G(q) = \begin{cases} p_{1,q}, & \text{if } q \in S,\\[4pt] p_{2,q}, & \text{if } q \notin S. \end{cases} \] When the next bit $ b \in \{0,1\} $ of the sequence $ X $ is revealed, $ G $ multiplies its current capital by $ 2 p_G(q)(b) $, updates the state of $ \mathcal{S} $ according to its transition on input $ b $, and simultaneously updates the internal state of $ G_{1} $ (if $ q \in S $) or $ G_{2} $ (if $ q \notin S $) according to their respective transition functions on the same input. Thus, at every step, $ G $ mirrors the run of the selector automaton, using $ G_{1} $ to place bets during selections and $ G_{2} $ otherwise. This construction ensures that the capital of $ G $ after reading the first $ n $ bits of $ X $ equals the product of the respective capital contributions from the segments on which $ \mathcal{S} $ selects and does not select.   Let $ D_G(X{\upharpoonright}n) $ denote the capital of $ G $ after reading the first $ n $ bits of $ X $. For all sufficiently large $ n $, the number of steps taken in selecting states is between $ (\lambda_{\mathcal{S}} - \varepsilon)n $ and $ (\lambda_{\mathcal{S}} + \varepsilon)n $, and the number of steps in non-selecting states is between $ ((1-\lambda_{\mathcal{S}}) - \varepsilon)n $ and $ ((1-\lambda_{\mathcal{S}}) + \varepsilon)n $. 
	Consequently, for infinitely many $ n $, $ D_G(X{\upharpoonright}n) \ge 2^{(1-s)(\lambda_{\mathcal{S}} - \varepsilon)n} \cdot 2^{(1-s')((1-\lambda_{\mathcal{S}}) - \varepsilon)n}. $ This gives \[ \frac{1}{n}\log D_G(X{\upharpoonright}n) \ge 1 - \lambda_{\mathcal{S}} s - (1-\lambda_{\mathcal{S}}) s' - 2\varepsilon. \] Hence, the martingale $ G $ succeeds on $ X $ with rate at least $ 2^{(1 - [\lambda_{\mathcal{S}} s + (1-\lambda_{\mathcal{S}}) s'] - 2\varepsilon)n} $. Since $ s > \dim_{\mathrm{FS}}(\mathcal{S}(X)) $ and $ s' > \Dim_{\mathrm{FS}}(\mathcal{S}^{c}(X)) $ were arbitrary, we obtain \[ \dim_{\mathrm{FS}}(X) \le \lambda_{\mathcal{S}}\cdot \dim_{\mathrm{FS}}(\mathcal{S}(X)) + (1-\lambda_{\mathcal{S}})\cdot\Dim_{\mathrm{FS}}(\mathcal{S}^{c}(X))+2\epsilon. \] On setting $2\epsilon$ to $\epsilon$, the proof is complete. \end{proof}
We next demonstrate that the classical result of Agafonov on normal numbers follows as a special case of our general theorem for finite-state dimension.

\begin{proof}[Proof of Theorem~\ref{thm:agafonovfornormal} using Theorem~\ref{thm:agafonovforfsd}]
	Let $ X $ be a normal number, so $ \dim_{\mathrm{FS}}(X)=1 $.
	We show that for every finite-state selector $ \mathcal{S} $, the subsequence $ \mathcal{S}(X) $ is normal.
	By Lemma~\ref{lem:irr_selector_equiv}, we may assume that $ \mathcal{S} $ is irreducible, so its Markov chain $ M_{\mathcal{S}} $ has a unique ergodic component. If all states in this component are non-selecting, then $ \mathcal{S}(X) $ is finite;  
	if all are selecting, $ \mathcal{S}(X) $ coincides with $ X $ after a finite prefix.  
	Hence we may assume $ 0<\lambda_{\mathcal{S}}<1 $, where  
	$ \lambda_{\mathcal{S}}=\sum_{q\in S}\pi_{\mathcal{S}}(q) $ is the stationary mass of selecting states. By Theorem~\ref{thm:agafonovforfsd}, for every $ \varepsilon>0 $ there exists $ \delta>0 $ such that
	$
	\lambda_{\mathcal{S}}\dim_{\mathrm{FS}}(\mathcal{S}(X))
	+(1-\lambda_{\mathcal{S}})\Dim_{\mathrm{FS}}(\mathcal{S}^{c}(X))
	\ge 1-\varepsilon.
	$
	Since $ 0<\lambda_{\mathcal{S}} <1 $ and both dimensions on the right are at most $1$,
	this inequality can hold for arbitrarily small $\varepsilon$ only if
	$\dim_{\mathrm{FS}}(\mathcal{S}(X))=1$.
	Thus every subsequence selected by a finite-state selector is normal.
	
	Conversely, if every such subsequence is normal, then the selector that chooses every bit
	yields $ X $ itself, implying that $ X $ is normal.
\end{proof}


In the proof of Theorem \ref{thm:agafonovforfsd}, the $\varepsilon$ term arises because, for a given selector~$\mathcal{S}$ and sequence~$X$, 
the limiting distribution of the run of $M_{\mathcal{S}}$ on $X$ need not exist—there may be multiple limit points in 
$\mathcal{W}_{M_{\mathcal{S}}}(X)$.  
However, if we restrict our attention to selectors whose induced Markov chains have a unique limiting distribution on the given sequence $X$, 
the argument immediately yields a sharper statement without the $\varepsilon$ term.  
We note this special case below.

\begin{theorem}
	\label{thm:special_case}
	Let $\mathcal{S}=(Q, q_0, \delta, S)$ be an irreducible selector with corresponding Markov chain $M_{\mathcal{S}}$, and let $X$ be a sequence such that 
	the empirical distribution of states in the run of $M_{\mathcal{S}}$ on $X$ converges to its stationary distribution $\pi_{\mathcal{S}}$. Then, \[
	\lambda_{\mathcal{S}}\dim_{\mathrm{FS}}(\mathcal{S}(X))
	+(1-\lambda_{\mathcal{S}})\Dim_{\mathrm{FS}}(\mathcal{S}^{c}(X))
	\ge \dim_{\mathrm{FS}}(X).
	\]
\end{theorem}

The proof follows directly from the argument of the main theorem, since the convergence assumption eliminates the need for 
the error term arising from the limiting-distribution approximation. The general results established above not only capture structural relationships between the finite-state dimensions of a sequence and its selected subsequences, but also serve as useful tools for deriving quantitative lower bounds on the finite-state dimension of sequences selected by finite automata.  The following lemmas illustrate such applications.
From Theorem \ref{thm:special_case}, we can obtain a simple but useful lower bound on the finite-state dimension of the subsequence selected by any selector.

\begin{lemma}
	\label{lem:lower_bound_selector}
	Let $\mathcal{S}=(Q, q_0, \delta, S)$ be an irreducible selector whose induced Markov chain $M_{\mathcal{S}}$ has stationary distribution $\pi_{\mathcal{S}}$, 
	and let $X$ be any sequence such that the empirical distribution of states during the run of $M_{\mathcal{S}}$ on $X$ 
	converges to $\pi_{\mathcal{S}}$.  
	Then,
	$
	\dim_{\mathrm{FS}}\!\bigl(\mathcal{S}(X)\bigr)
	\;\ge\;
	\frac{1}{\lambda_{\mathcal{S}}}
	\Bigl(\,
	\dim_{\mathrm{FS}}(X)
	- (1-\lambda_{\mathcal{S}})
	\,\Bigr)
	$
	where $\lambda_S$ is again the stationary mass of the selecting states. 
\end{lemma}

\begin{proof}
	From Theorem~\ref{thm:special_case},
	$
	\lambda_{\mathcal{S}} \dim_{\mathrm{FS}}(\mathcal{S}(X))
	+ (1-\lambda_{\mathcal{S}}) \Dim_{\mathrm{FS}}(\mathcal{S}^{c}(X))
	\ge \dim_{\mathrm{FS}}(X).
	$
	Since $\Dim_{\mathrm{FS}}(\mathcal{S}^{c}(X)) \le 1$, we have
	$
	\lambda_{\mathcal{S}} \dim_{\mathrm{FS}}(\mathcal{S}(X))
	\ge \dim_{\mathrm{FS}}(X) - (1-\lambda_{\mathcal{S}}),
	$
	and dividing both sides by $\lambda_{\mathcal{S}}$ gives the desired inequality.
\end{proof}

As an application, we observe that Lemma~17 from \cite{Nandakumar2024}, which established a lower bound on the finite-state dimension of arithmetic progression subsequences, is a special case of Lemma~\ref{lem:lower_bound_selector}. Before doing so, we clarify the notation. For any sequence $X = x_0 x_1 x_2 \ldots$ over $\{0,1\}$ and integers $d \ge 1$ and $0 \le j < d$, define $X^{(d,j)} = x_j\,x_{j+d}\,x_{j+2d}\,x_{j+3d}\,\ldots$, the \emph{$j$-th arithmetic progression subsequence} of $X$ with common difference $d$. Each such subsequence can be selected by a simple $d$-state selector automaton $\mathcal{S}_{d,j}$ whose states form a directed cycle of length $d$. Starting from position $0$, the automaton moves deterministically from state $i$ to $(i+1) \bmod d$ at each step, selecting exactly when it enters state $j$. The corresponding Markov chain $M_{\mathcal{S}_{d,j}}$ is therefore a $d$-cycle, which is irreducible with stationary distribution uniform over the $d$ states, that is, $\pi_{\mathcal{S}_{d,j}}(q) = \tfrac{1}{d}$ for all $q$, so that $\lambda_{\mathcal{S}_{d,j}} = \tfrac{1}{d}$.

\begin{lemma}[\cite{Nandakumar2024}]
	\label{lem:arith_prog}
	For every infinite sequence $X \in \{0,1\}^\infty$, integer $d \ge 1$, and $j \in \{0,1,\ldots,d-1\}$,
	$
	\dim_{\mathrm{FS}}(X^{(d,j)})
	\;\ge\;
	d\!\left(
	\dim_{\mathrm{FS}}(X) - \frac{d-1}{d}
	\right).
	$
\end{lemma}

\begin{proof}
	As noted above, the selector $\mathcal{S}_{d,j}$ defining the $j$th arithmetic progression subsequence 
	has $\lambda_{\mathcal{S}_{d,j}} = 1/d$ and a uniform stationary distribution that the empirical state frequencies 
	converge to for every input sequence $X$.  
	Hence the hypotheses of Lemma~\ref{lem:lower_bound_selector} are satisfied.  
	Applying Lemma~\ref{lem:lower_bound_selector} with this selector yields the desired inequality.
\end{proof}

\subsection{Necessity and Tightness of the Main Bound}

\begin{figure}[h!]
	\centering
	\begin{tikzpicture}[shorten >=1pt, node distance=2cm, on grid, auto]
		\node[state, initial] (a) {$a$};
		\node[state] (b) [right=of a] {$b$};
		\node[state] (d) [below=of a] {$d$};
		\node[state] (c) [below=of b] {$c$};
		
		\path[->]
		(a) edge[bend left=15] node {0/1} (b)
		(b) edge[bend left=15] node {0} (a)
		(c) edge[bend left=15] node {0/1} (d)
		(d) edge[bend left=15] node {1} (c)
		(b) edge[bend left=10] node[pos=0.3, right] {1} (c)
		(d) edge[bend left=10] node[pos=0.3, left] {0} (a);
	\end{tikzpicture}
	\caption{%
		Base automaton used to define the selectors in the examples below.
		Each example corresponds to a distinct choice of selecting states within this automaton.
	}
	\label{fig:selector-automaton}
\end{figure}
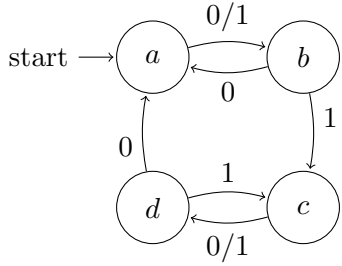

We now examine the role of the assumptions and sharpness of the main theorem by means of concrete examples.  
All the selectors considered in this subsection are derived from the basic automaton shown in Figure~\ref{fig:selector-automaton}.  
These examples demonstrate three key aspects of the theorem:  
(i)~the necessity of the condition $ \dim_{\mathrm{FS}}(X) \ge 1 - \delta $,  
(ii)~the tightness of the inequality, and  
(iii)~the possibility of strict inequality. Fix a normal sequence with binary expansion $r_1 r_2 r_3 r_4 \ldots$.  
The \emph{diluted sequence} obtained by inserting a $0$ in every alternate position:  
$r_1 0 r_2 0 r_3 0 r_4 0 \ldots$ has finite-state dimension $ \tfrac{1}{2} $ \cite{Dai2004}.

\paragraph{Example 1 (Necessity of the $1-\delta$ condition).}
We show that the assumption $ \dim_{\mathrm{FS}}(X) \ge 1 - \delta $ is essential by constructing a sequence $ X $ and selector $ \mathcal{S}_1 $ for which the inequality in the main theorem fails.  
Let $ \mathcal{S}_1 $ be obtained from the automaton in Figure~\ref{fig:selector-automaton} by designating only state $ a $ as the selecting state.  
When run on the diluted sequence, the automaton alternates between $ a $ and $ b $, always selecting the digits in odd positions, which form the original normal sequence.  
Hence, $ \lambda_{\mathcal{S}_1} = \tfrac{1}{4} $, $ \dim_{\mathrm{FS}}(\mathcal{S}_1(X)) = 1 $, and $ \Dim_{\mathrm{FS}}(\mathcal{S}_1^{c}(X)) = 0 $.  
The left-hand side of the inequality equals $ \tfrac{1}{4} $, while $ \dim_{\mathrm{FS}}(X) = \tfrac{1}{2} $.  
Thus,
$
\lambda_{\mathcal{S}_1}\dim_{\mathrm{FS}}(\mathcal{S}_1(X))
+(1-\lambda_{\mathcal{S}_1})\Dim_{\mathrm{FS}}(\mathcal{S}_1^{c}(X))
< \dim_{\mathrm{FS}}(X),
$
showing that the $1-\delta$ requirement cannot be omitted.

\paragraph{Example 2 (Tightness).}
Let $ \mathcal{S}_2 $ have selecting states $ a $ and $ c $, so that $ \lambda_{\mathcal{S}_2} = \tfrac{1}{2} $.  
Running on the diluted sequence, this selector outputs the normal sequence $ r_1 r_2 r_3 \ldots $, while its complement selects only zeros.  
Hence $ \dim_{\mathrm{FS}}(\mathcal{S}_2(X)) = 1 $ and $ \Dim_{\mathrm{FS}}(\mathcal{S}_2^{c}(X)) = 0 $, so that
$
\lambda_{\mathcal{S}_2}\dim_{\mathrm{FS}}(\mathcal{S}_2(X))
+(1-\lambda_{\mathcal{S}_2})\Dim_{\mathrm{FS}}(\mathcal{S}_2^{c}(X))
= \dim_{\mathrm{FS}}(X),
$
demonstrating that the bound can be tight.

\paragraph{Example 3 (Strict inequality).}
Finally, let $ \mathcal{S}_3 $ select in states $ a, c, d $, giving $ \lambda_{\mathcal{S}_3} = \tfrac{3}{4} $.  
As before, $ \dim_{\mathrm{FS}}(\mathcal{S}_3(X)) = 1 $ and $ \Dim_{\mathrm{FS}}(\mathcal{S}_3^{c}(X)) = 0 $, and therefore we obtain
$
\lambda_{\mathcal{S}_3}\dim_{\mathrm{FS}}(\mathcal{S}_3(X))
+(1-\lambda_{\mathcal{S}_3})\Dim_{\mathrm{FS}}(\mathcal{S}_3^{c}(X))
=\tfrac{3}{4}
> \tfrac{1}{2}
=\dim_{\mathrm{FS}}(X),
$
demonstrating that the inequality can be strict.

\bibliography{main}
\bibliographystyle{plain}

\end{document}